\documentclass{article}
\usepackage{amsmath}
\usepackage{amssymb}
\usepackage{amsthm}
\usepackage[utf8]{inputenc}
\usepackage{authblk}

 \DeclareMathOperator{\diag }{diag}

\newtheorem{thm}{Theorem}
\newtheorem{lem}{Lemma}

\title{Future of Bianchi I magnetic cosmologies with kinetic matter}

\author[1]{Ho Lee\footnote{holee@khu.ac.kr}}
\author[2]{Ernesto Nungesser\footnote{em.nungesser@upm.es}}
\affil[1]{Department of Mathematics and Research Institute for Basic Science, College of Sciences, Kyung Hee University, Seoul 02447, Republic of Korea}
\affil[2]{M2ASAI, Universidad Polit\'{e}cnica de Madrid, ETSI Navales, Avda. de la Memoria, 4, 28040 Madrid, Spain}

\begin{document}
\maketitle

\begin{abstract}
We show under the assumption of small data that solutions to the Einstein-Vlasov system with a pure magnetic field and Bianchi I symmetry isotropise and tend to dust solutions. We also obtain the decay rates for the main variables. This generalises part of the work [V.~G.~LeBlanc, Classical Quantum Gravity 14, 2281-2301 (1997)] concerning the future behaviour of orthogonal perfect fluids with a linear equation of state in the presence of a magnetic field to the Vlasov case. 
\end{abstract}

\section{Introduction}

The $\Lambda$CDM-model is very successful in describing most cosmological phenomena. However in recent years discrepancies related to key cosmological parameters have arisen which is sometimes also called the Hubble tension, due to discrepancies concerning the value of the Hubble constant. How to solve this tension is unclear, because there exist a vast amount of possibilities, but no preferred alternative. For a recent review we refer to \cite{DV}. Among the possible solutions are the existence of primordial magnetic fields (cf.~Section 12.4 of \cite{DV} and \cite{Jedamzik}) or the replacement of the spatially flat FLRW metric with the Bianchi type-I metric (cf.~Section 14.2 of \cite{DV} and \cite{Akarsu}).

Here we will consider the late time asymptotics of Bianchi I solutions to the Einstein-Vlasov system with a pure magnetic field assuming small data. More precisely we will assume that the dispersion of the momenta of the particles, the anisotropy of the space-time and the magnetic field are small. Note that we do not assume the presence of charged particles, but just a magnetic field. For some interesting results concerning the former case we refer to \cite{BF, NT, Tchapnda}. For more background on the interest of primordial magnetic fields we refer to \cite{Jedamzik} and references therein.

A magnetic field makes the space-time necessarily anisotropic and it turns out that the existence of a source-free magnetic field restricts the Bianchi types to be of type I, II, III, VI$_0$ or VII$_0$ (cf.~\cite{LB2} and references therein). The late time asymptotics for these space-times with a magnetic field and a non-tilted perfect fluid has been studied in \cite{C, HW, LB, LB2, LKW}. 

Here we replace the non-tilted perfect fluid with collisionless matter and study its long time behaviour. It is very often the case that additional simplifications are assumed such as Bianchi I LRS symmetry, so that the metric is diagonal and has only two different components or that the magnetic field is aligned along a shear eigenvector. Here we do not make these simplifications and follow the framework of \cite{LB} where without loss of generality the magnetic field is aligned along one of the frame vectors. Using the methods developed in \cite{EN} (see also \cite{HU} for the reflection symmetric case) we obtain that Bianchi I solutions to the Einstein-Vlasov system with a magnetic field tend to a FLRW solution with dust. The magnetic field which is present in the beginning vanishes asymptotically. We have thus generalised the results of LeBlanc \cite{LB} to the Vlasov case in the case of Bianchi I under the assumption of small data.  For completeness we also consider the case of a positive cosmological constant obtaining a de Sitter like behaviour following Hayoung Lee \cite{HL}, see also \cite{Ayissi, NT,Tchapnda}. 

One conclusion of this is that one might add to the standard cosmological model a small magnetic field, replace the FLRW metric by a non-diagonal Bianchi I metric such that the shear is small or replace a perfect fluid model by collisionless matter with small dispersion of the momenta, obtaining a similar qualitative asymptotic behaviour as without these extensions. However the additional degrees of freedom might be used to relieve the Hubble tension. We refer to \cite{DV} and references therein for different approaches in that direction.

The structure of the paper is as follows. In Section \ref{general} we briefly present the general equations and introduce some notations. For more details we refer to \cite{Ring,WE}. In Section \ref{bianchiI} we impose Bianchi I symmetry and express the Maxwell, Vlasov and Einstein equations and all the relevant quantities using an orthonormal frame as in \cite{LB,WE}. The late time asymptotics are obtained in Section \ref{massive}. Finally in Section \ref{lambda} we consider the late time behaviour in presence of a positive cosmological constant.

\section{The Einstein-Vlasov system with a pure magnetic field} \label{general}
In the following Greek indices will denote space-time indices and run from 0 to 3, while Latin indices denote space-like indices and run from 1 to 3. 

Let us consider the source free Einstein-Vlasov-Maxwell equations in a space-time with metric $g_{\alpha\beta}$: 
\begin{align}
\label{einstein}&G_{\alpha\beta}=T_{\alpha\beta}= T_{\alpha\beta}^{\mathrm{Vl}}+ T_{\alpha\beta}^{\mathrm{M}},\\
\label{maxwell1}&\nabla_\beta F^{\alpha\beta} =0, \\
\label{maxwell2} &\nabla_{[\delta} F_{\alpha\beta]}=0,\\
\label{vlasov} &p^\alpha \frac{\partial f}{\partial x^\alpha}-\Gamma^a_{\beta \gamma} p^\beta p^\gamma \frac{\partial f}{\partial p^a}= 0,
\end{align}
with
\begin{align}\label{TVl}
T_{\alpha\beta}^{\mathrm{Vl}}=  \int \chi p_{\alpha} p_{\beta},
\end{align}
where the integration is over the future pointing mass-shell at a given space-time point which is defined by
\begin{align}\label{massshell}
p_{\alpha} p_{\beta}  \,  { g}^{\alpha\beta}=- m^2, \quad p^0>0,
\end{align}
where $m \geq 0 $, $\chi$ is the distribution function multiplied by the Lorentz invariant measure and
\begin{align}\label{magnetict}
T_{\alpha\beta}^{\mathrm{M}}= F_{\alpha\gamma} {F_{\beta}}^\gamma-\frac14 g_{\alpha\beta} F_{\gamma\delta} F^{\gamma\delta}.
\end{align}
Equation \eqref{einstein} is the Einstein equation where $G_{\alpha\beta}$ is the Einstein tensor which is equal to the energy momentum tensor $T_{\alpha\beta}$ which is split into the part coming from the matter described by kinetic theory $T_{\alpha\beta}^{\mathrm{Vl}}$ and the part coming from the Maxwell field strength $F_{\alpha\beta}$ which gives rise to $T_{\alpha\beta}^{\mathrm{M}}$. 

From \eqref{magnetict} we can see that $T_{\alpha\beta}^{\mathrm{M}}$ is trace-free. Equations \eqref{maxwell1}--\eqref{maxwell2} are the Maxwell equations for the Maxwell field tensor $F_{\alpha\beta}$ without sources, so in particular we are considering that the particles are not charged.

Finally we have the Vlasov equation \eqref{vlasov} for the distribution function $f(x^\alpha, p^b)$ where $p^\alpha$ are the momenta of the particles on the mass-shell $p^\alpha p_\alpha= -m^2$, where $p^0$ is expressed via the mass-shell condition in terms of the space-like components of the momenta. For simplicity we will assume that the initial distribution function has compact support in momentum space (away from the origin in the massless case).


\section{The Bianchi I symmetric Einstein-Vlasov system with a pure magnetic field} \label{bianchiI}
\subsection{Bianchi I spacetimes in an orthonormal frame approach}

We will use an orthonormal frame approach as in \cite{LB} and thus the main variables are the commutator functions $\gamma^{\delta}_{\alpha\beta}$ given by
\begin{align}
[e_\alpha, e_\beta]= \gamma^{\delta}_{\alpha\beta}e_\delta,
\end{align}
where $\{e_\alpha \}$ are the frame vectors which are orthonormal.

For Bianchi I the only non-vanishing commutator functions are 
\begin{align}\label{commutators}
\gamma^{a}_{0b}= - \gamma^{a}_{b0} = -{ \theta^{a}}_{b}- {\epsilon^{a}}_{bc}\Omega^c,
\end{align}
where $\theta_{ab}$ are the frame components of the rate of expansion tensor, $\epsilon_{abc}$ denotes the alternating
symbol with $\epsilon_{abc}=1$ and $\Omega^a$ are the components of the local angular velocity of the spatial frame $\{e _ a\}$ with respect to a Fermi-propagated spatial frame.

We express the frame components of the rate of expansion tensor in terms of the rate of shear tensor $\sigma_{ab}$ and the rate of volume expansion scalar $\theta=\theta^a_a$:
\begin{align}\label{sigma}
\sigma_{ab}= \theta_{ab}-\frac13 \theta \delta_{ab},
\end{align}
where $\delta_{ab}$ is the usual Kronecker symbol.

The connection coefficients for an orthonormal frame with 
\begin{align}\label{ortho}
g_{\alpha\beta}= \eta_{\alpha \beta} = \diag (-1,1,1,1)
\end{align} 
are given as
\begin{align}\label{connections}
\Gamma_{\beta\gamma}^{\delta}=\frac12 \eta^{\alpha\delta} (\gamma^\epsilon_{\gamma\beta} \eta_{\alpha\epsilon}
+ \gamma^\epsilon_{\alpha\gamma} \eta_{\beta\epsilon} - \gamma^\epsilon_{\beta\alpha} \eta_{\gamma\epsilon}) ,
\end{align}
(cf.~(1.14) and (1.59) of \cite{WE}). As a consequence of \eqref{commutators} and \eqref{connections} the only non-vanishing connection coefficients are
\begin{align}\label{connection}
\Gamma_{bc}^{0} =\frac12 \left( \gamma^b_{c0}+\gamma^c_{b0} \right), \quad
 \Gamma_{0b}^{d}= \frac12  \left(\gamma^d_{b0}  + \gamma^b_{ d0} \right), \quad
  \Gamma_{b0}^{d}=  \frac12 \left( \gamma^d_{0b}+\gamma^b_{d0} \right).
\end{align}

Since we use an orthonormal frame upper and lower spatial indices are in principal not distinguished. However to increase readability we will use the Einstein summation convention that repeated indices (one lower and one upper index) will be summed over.

\subsection{Maxwell equations}

For a pure magnetic field relative to the invariant frame we have:
\begin{align*}
h_\alpha= (0,h_a),
\end{align*}
and the Maxwell equations can be simplified to (cf.~(2.4) of \cite{LKW}):
\begin{align*}
\dot{h}_a= -\frac23 \theta h_a + \sigma_{ab} h^b + \epsilon_{abc}\Omega^c h^b,
\end{align*}
where a dot means derivation with respect to time $t$.
Now for simplicity we align the frame vector $e_1$ as in \cite{LB} along the magnetic field, so that 
\begin{align}\label{h}
h_\alpha=(0,h_1,0,0).
\end{align}
Later in the end of Section \ref{hubblenormalised} we will see that if the magnetic field is non-zero initially it will remain non-zero (cf.~\eqref{blower}).

As a consequence from \eqref{h} the Maxwell equations are
\begin{align}\label{maxwellBianchiI}
\dot{h}_1 =  \left(-\frac23 \theta + \sigma_{11}\right)h^1,  \quad \dot{h}_2 &= \left(\sigma_{21}- \Omega^3\right) h^1,  \quad \dot{h}_3 =  \left(\sigma_{31}+ \Omega^2\right)h^1.
\end{align}
In order to have the magnetic field along $e_1$ for all time, we must have 
\begin{align} \label{maxconstraint}
\sigma_{21}- \Omega^3 = 0,\quad \sigma_{31}+ \Omega^2 = 0.
\end{align}
There is no restriction for $\Omega^1$, but for simplicity we choose it to vanish.  Let us from now on use the notation $h_1=h$.

For a pure magnetic field we have 
\begin{align*}
F_{\alpha \beta}={ \epsilon_{\alpha \beta}}^{\gamma \delta} h_{\gamma} u_{\delta},
\end{align*}
where $\epsilon^{\alpha \beta\gamma \delta}$ is the totally antisymmetric permutation symbol (we use the convention that $\epsilon^{0123}=1$, cf.~page 16 of \cite{WE}) and $u^{\alpha}$ is the unit timelike vector field given by $e_0$. As a consequence of \eqref{h} we have that the only non-vanishing components of $F_{\alpha \beta}$ are:
\begin{align}\label{faraday}
F_{23}= h = - F_{32}.
\end{align}
This means, using \eqref{magnetict}, \eqref{ortho} and \eqref{faraday}
\begin{align*}
T_{\alpha\beta}^{\mathrm{M}} = \frac12 h^2 \diag (1,-1,1,1).
\end{align*}
It follows from the standard decomposition of the energy-momentum tensor (cf.~(1.19) of \cite {WE}) with $u^\alpha=e_0= (1,0,0,0)$ that the non-vanishing terms are
\begin{align*}
& \rho^{\mathrm{M}} = \frac12 h^2, \\
& \mathcal{P}^{\mathrm{M}} = \frac13 \rho^{\mathrm{M}}=\frac16  h^2 ,\\
& \pi^{\mathrm{M}}_{\alpha \beta} = \frac13 h^2 \diag (0, -2, 1, 1).
\end{align*}

\subsection{The Vlasov equation with Bianchi I symmetry}
 For the Vlasov equation \eqref{vlasov} we need to compute the contraction of the connection coefficients with $p^\beta p^\gamma$. Considering the non-vanishing components of the connection \eqref{connection} we obtain
 \begin{align*}
 \Gamma^a_{\beta \gamma} p^\beta p^\gamma &= \Gamma^a_{0c} p^0 p^c + \Gamma^a_{b 0} p^b p^0 = \frac12p^0  \left(\gamma^a_{c0} p^c  + \gamma^c_{ a0} p_c   +
   \gamma^a_{0b}p^b +\gamma^b_{a0} p_b \right)   \\
  &=  p^0 p_c \gamma^c_{ a0} = p^0 p_c \left( {\theta^{c}}_{a} +{\epsilon^{c}}_{ad}\Omega^d \right).
 \end{align*}
As a consequence the Vlasov equation \eqref{vlasov} using a group invariant frame such that $f$ does not depend on $x^a$ turns into
\begin{align}\label{vlasovBianchiI}
 \frac{\partial f}{\partial t}- p_c \left( {\theta^{c}}_{a}+ {\epsilon^{c}}_{ad}\Omega^d \right) \frac{\partial f}{\partial p^a}= 0.
\end{align}
This is a partial differential equation which is of first order. We thus can use the method of characteristics (cf.\ Chapters 1.4 and 1.5 of \cite{John}) which will be done as follows.

\subsubsection{The characteristic system}

Applying the method of characteristics, we see that the distribution function is constant along the characteristics and the parameter of the characteristic system can be chosen to be proper time. We define the characteristic curve $ P^a(t)$ by
\begin{align}\label{charac}
\dot{ P }^a=  - P_c \left( {\theta^{c}}_{a}+ {\epsilon^{c}}_{ad}\Omega^d \right),
\end{align}
where $P^a(t)=p^a$ given $t$. We use the notation $P^a$ to indicate that $p^a$ is parametrised by $t$ and call this the characteristic momenta. Then, consider the quantity
\begin{align*}
P = P^a P_a .
\end{align*}
It has the following evolution equation
\begin{align}
\nonumber \dot{P} &=   -2 P^a P_c \left( {\theta^{c}}_{a}+ {\epsilon^{c}}_{ad}\Omega^d \right)= -2 P^a P_c  {\theta^{c}}_{a}= -2 P^a P_c (\sigma^c_{a}+\frac13 \theta \delta^c_{a}) \\
\label{P} &=-\frac23 \theta P -2 P^a P_c \sigma^c_{a},
\end{align}
where we have used the symmetry of $P^a P_c$ together with the antisymmetry of ${\epsilon^{c}}_{ad}$ and \eqref{sigma}.

\subsubsection{The energy momentum tensor of the kinetic part}

From the kinetic part written in an orthonormal frame we have from \eqref{TVl}:
\begin{align}\label{TVlhomo}
T^{\mathrm{Vl}}_{\alpha\beta} = \int f(t,p) \frac{p_{\alpha}p_{\beta}}{p^0} dp , 
\end{align}
where the integration is over the future pointing mass-shell at a given space-time point, and we used the notations $ p = (p_1,p_2,p_3)$ and $dp=dp_1dp_2dp_3$. In particular, we have
\begin{align}
& \label{defrho} \rho^{\mathrm{Vl}} =T_{00} ^{\mathrm{Vl}} = \int f(t,p) p^0 dp,\\
& q_a^{\mathrm{Vl}} = - T_{0a} ^{\mathrm{Vl}}=  \int f(t,p) p_a dp,\\
& \mathcal{P}^{\mathrm{Vl}} = \frac13 S,\quad \pi^{\mathrm{Vl}}_{ab} = T_{ab}^{\mathrm{Vl}} - \frac13 S \delta_{ab},
\end{align}
with
\begin{align}
\label{defS} S= T_{ab}^{\mathrm{Vl}} \delta ^{ab}=  T_{11}^{\mathrm{Vl}}+ T_{22}^{\mathrm{Vl}}+ T_{33}^{\mathrm{Vl}} .
\end{align}
Note that in our frame \eqref{massshell} means that 
\begin{align*}
(p^0)^2 = p_a p^a + m^2,
\end{align*}
which implies $p_a p^a \leq (p^0)^2$. In particular with the definitions just introduced we obtain for any $a, b$ that
\begin{align}\label{energyconditions}
\vert \pi_{ab}^{\mathrm{Vl}} \vert \leq S \leq \rho^{\mathrm{Vl}}.
\end{align}

\subsection{The Einstein field equations}
Let us introduce the usual Hubble variable:
\begin{align*}
H=\frac13 \theta .
\end{align*}
Then, the Einstein field equations are given by (1.90)--(1.93) of \cite{WE}, which reduce in our case to:
\begin{align}
\label{evH} &\dot{H}= -H^2 - \frac13  \sigma_{ab}\sigma^{ab}- \frac16 (\rho + 3\mathcal{P}),\\
&\dot{\sigma}_{22}=-3H \sigma_{22}-2\sigma_{12}^2 +\frac13 h^2 + \pi_{22}^{\mathrm{Vl}},\\
&\dot{\sigma}_{33}=-3H \sigma_{33}-2\sigma_{13}^2+\frac13 h^2 + \pi_{33}^{\mathrm{Vl}},\\
&\dot{\sigma}_{23}= -3H\sigma_{23}-2\sigma_{12}\sigma_{13}+\pi_{23}^{\mathrm{Vl}},\\
&\dot{\sigma}_{12}=-3H\sigma_{12}+2\sigma_{12}\sigma_{22}+\sigma_{12} \sigma_{33}+\sigma_{13} \sigma_{23}+\pi_{12}^{\mathrm{Vl}},\\
&\dot{\sigma}_{13}=-3H\sigma_{13}+2\sigma_{13}\sigma_{33}+\sigma_{13} \sigma_{22}+\sigma_{12} \sigma_{23}+\pi_{13}^{\mathrm{Vl}},\\
\label{constraintrho}& \rho= \rho^{\mathrm{Vl}}+\rho^{\mathrm{M}} =\rho^{\mathrm{Vl}}+\frac12h^2= 3H^2-\frac12 \sigma_{ab}\sigma^{ab},\\
\label{momentumconstraint}&q_a= q_a^{\mathrm{Vl}} =0 ,
\end{align}
where we have used the constraint equations \eqref{maxconstraint}, the choice $\Omega^1=0$ and the fact that $\pi_{\alpha \beta}= \pi_{\alpha \beta}^{\mathrm{Vl}}+\pi^{\mathrm{M}}_{\alpha \beta}$ and $\sigma_{11}= - \sigma_{22}-\sigma_{33}$.

Now we use the constraint \eqref{constraintrho} to eliminate the energy density in \eqref{evH} and substitute $\mathcal{P}=\mathcal{P}^{\mathrm{M}}+\mathcal{P}^{\mathrm{Vl}}$. As a result we obtain
\begin{align*}
-\frac{\dot{H}}{H^2} = \frac32 + \frac{1}{4} \frac{\sigma_{ab}\sigma^{ab}}{H^2} + \frac16 \left (\frac{S}{H^2}+ \frac12 \frac{h^2}{H^2} \right).
\end{align*}

\subsection{The Hubble normalised Bianchi I symmetric Einstein-Vlasov system with a pure magnetic field} \label{hubblenormalised}
The full system consists of the Maxwell \eqref{maxwellBianchiI}--\eqref{maxconstraint}, the Vlasov \eqref{vlasovBianchiI}, and the Einstein field equations \eqref{evH}--\eqref{momentumconstraint}. We consider the Hubble normalised variables
\begin{align*}
 \Sigma_{ab}= \frac{\sigma_{ab}}{H}, \quad B= \frac{h}{H}, \quad \Pi_{ij}= \frac{\pi_{ij}^\mathrm{Vl}}{H^2},
 \end{align*} 
and the dimensionless time variable $\tau$ defined by
\begin{align}\label{dtdtau}
\frac{dt}{d\tau}=\frac{1}{H}.
\end{align} 
Introducing  the deceleration parameter $q$
\begin{align}\label{q}
q = -1-\frac{\dot{H}}{H^2},
\end{align} 
the full system of equations is
\begin{align}
 &\label{H-1}-\frac{\dot{H}}{H^2} = \frac32 + \frac{1}{4} \Sigma_{ab}\Sigma^{ab} + \frac16 \left (\frac{S}{H^2}+ \frac12 B^2 \right),\\
& \label{B} B'=(q-1-\Sigma_{22}-\Sigma_{33})B,\\
&\label{22} \Sigma'_{22}= (q-2) \Sigma_{22}-2\Sigma_{12}^2+\frac13 B^2 + \Pi_{22},\\
&\label{33}\Sigma'_{33}=(q-2) \Sigma_{33}-2\Sigma_{13}^2+\frac13 B^2 + \Pi_{33} , \\
&\label{23}\Sigma'_{23}= (q-2)\Sigma_{23}-2\Sigma_{12}\Sigma_{13}+\Pi_{23},\\
&\label{12} \Sigma'_{12}=(q-2)\Sigma_{12}+2\Sigma_{12}\Sigma_{22}+\Sigma_{12} \Sigma_{33}+\Sigma_{13} \Sigma_{23}+\Pi_{12},\\
&\label{13} \Sigma'_{13}=(q-2)\Sigma_{13}+2\Sigma_{13}\Sigma_{33}+\Sigma_{13} \Sigma_{22}+\Sigma_{12} \Sigma_{23}+\Pi_{13},\\
 \label{vlasovtau}& f'- p_c \left( \Sigma^{c}_{a}+\delta^{c}_{a}+\frac{{\epsilon^{c}}_{ad}\Omega^d}{H} \right) \frac{\partial f}{\partial p^a}= 0,\\
&\label{einsteinconstraints} \frac{\rho}{H^2}= \frac{\rho^{\mathrm{Vl}}}{H^2}+\frac12B^2= 3-\frac12 \Sigma_{ab}\Sigma^{ab}, \quad q_a^{\mathrm{Vl}} =0,\\
&\label{maxconstraints}\Sigma_{21}= \frac{\Omega^3}{H},\quad \Sigma_{31}=-\frac{ \Omega^2}{H}, \quad \frac{\Omega^1}{H}=0,
\end{align}
where a tilde denotes differentiation with respect to $\tau$.

\subsubsection{A first lower bound on the magnetic field}
From \eqref{einsteinconstraints} we have that
\begin{align*}
\frac{\rho}{H^2}+ \frac12 \Sigma_{ab}\Sigma^{ab}= \frac{\rho^{\mathrm{Vl}}}{H^2}+\frac12B^2+\frac12 \Sigma_{ab}\Sigma^{ab}= 3 .
\end{align*}
By definition $\rho$ and $\rho^{\mathrm{Vl}}$ are non-negative, thus we have that  $\frac{\rho}{H^2}$, $\frac{\rho^{\mathrm{Vl}}}{H^2}$,  any component of $\Sigma_{ab}$ and $B$ are bounded from above and below. From \eqref{energyconditions} we have that $\frac{S}{H^2}$ is bounded from above and by the definition \eqref{TVlhomo} with \eqref{defS} of $S$ it is also bounded from below since it is non-negative. Putting these bounds together and using \eqref{H-1} we can see from the definition \eqref{q} of $q$ that $q$ is also bounded from above and below. From \eqref{B} using a Gr{\"o}nwall type argument we can conclude that there exists a positive constant $C$ such that
\begin{align}\label{blower}
\vert B \vert \geq \vert B(\tau_0) \vert \exp [- C (\tau-\tau_0)],
\end{align}
which implies that if $B(\tau_0)$ is non-zero initially it will remain non-zero. One could also come to the conclusion that it has to remain non-zero using a local uniqueness argument together with the fact that equation \eqref{B} is homogeneous in $B$. The important thing here is that the frame vector $e_1$ defined along the magnetic field \eqref{h} will thus be always well-defined. We will obtain a better lower bound in the next section using some smallness assumptions.

\section{Future asymptotics} \label{massive}
We will now consider the massive case in order to generalise \cite{LB, EN}. For simplicity we assume that all particles have the same mass which we normalise to one. In \cite{LB} a pure magnetic field was considered with a perfect fluid. In \cite{EN} the Vlasov case with massive particles was treated without a magnetic field.  As already mentioned we are following the notation of \cite{LB} and in particular are using an orthonormal frame. The proof of our main result concerning massive particles however will follow the arguments of \cite{LN2,EN}. 

We will assume that the universe is expanding which in the case of small shear usually leads to a decrease in the dispersion of the momenta of the particles. The dispersion of the particles is an upper bound for $\frac{S}{ \rho^{\mathrm{Vl}}}$ (cf.\ \eqref{supportbound} below), which means that if the dispersion of the particles decreases, there is a dust-like behaviour. As a consequence it is reasonable to expect that the solution to the Einstein-Vlasov system with a small dispersion of the momenta should behave similarly to the solution of the Einstein-dust system which has already been studied in \cite{LB} (considering $\gamma=1$ which corresponds to dust). 

In order to actually prove this, we will use a bootstrap argument, which is an analogue of mathematical induction where instead of natural numbers non-negative real numbers are used (cf.\ Section 10.3 of \cite{Rendall}). The expectation is that the decay rates are like in the dust case or similar. So we will start with some weaker decay rates in the hope that using the evolution equations we can close the bootstrap argument, which will be done in Section \ref{main} in the first part of Lemma \ref{lem1}. Before that we will obtain some relations between the dispersion of the momenta and related quantities.

\subsection{Dispersion of the momenta and related quantities}
In \cite{LN2,EN} massive Vlasov particles were considered and it was shown that the dispersion of the momenta of the particles tends to zero if the shear is small. From \eqref{P} we have that
\begin{align}\label{support}
P'= -2 P -2 P^a P_c \Sigma^c_{a}.
\end{align}
This means in particular that the evolution of the dispersion of the particles does only depend on the shear of the space-time. The Vlasov equation does not depend explicitly on the magnetic field and the evolution of the support does also not depend explicitly on the local angular velocity with respect to a Fermi-propagated spatial frame. In fact the part depending on it vanished while deducing \eqref{P}. As a consequence we obtain again a similar result to Proposition 1 and Corollary 1 of \cite{LN2}. From \eqref{support} we have
\begin{align}\label{ineqP}
\frac{P'}{P} \leq -2 + 2 \sqrt{\Sigma_{ab}\Sigma^{ab}},
\end{align}
which implies that
\begin{align}\label{estimateP}
P(\tau) \leq P(\tau_0) \exp \int^{\tau}_{\tau_0} \left (-2 + 2 \sqrt{\Sigma_{ab}\Sigma^{ab}}\right ) ds.
\end{align}

To obtain our main result concerning massive particles we will assume smallness of the anisotropy, the magnetic field and the dispersion of the momenta, i.e.~we assume that  $\vert \Sigma_{ab}(\tau_0)\vert $, $\vert B(\tau_0)\vert $ and $ \hat{P}(\tau_0)$ are small where $\hat{P}$ is defined as
\begin{align}\label{defPhat}
\hat{P}(\tau)= \sup\{ \vert p \vert ^2 =\delta_{ab}p^ap^b : f(\tau,p)\neq 0 \}.
\end{align} 
Then, we use $\hat{P}$ to estimate $S$ as follows:
\begin{align}\label{supportbound}
S = \int f(\tau ,p)  \frac{p_1^2 + p_2^2 + p_3^2}{p^0} dp \leq \hat{P}  \int f(\tau ,p)  \frac{1}{p^0} dp \leq  \hat{P} \rho^{\mathrm{Vl}},
\end{align} 
since in the massive case $p^0 \geq 1$. Using \eqref{einsteinconstraints} we have $\frac{\rho^{\mathrm{Vl}}}{H^2}\leq 3$ so that
 \begin{align}\label{boundSH2}
  \frac{S}{H^2} \leq 3 \hat{P}.
 \end{align} 
This implies that we can also bound $\Pi_{ij}$ due to \eqref{energyconditions}
 \begin{align}\label{Pi}
\vert \Pi_{ij}\vert =\frac{1}{H^2}\vert \pi_{ij}^{\mathrm{Vl}}\vert \leq \frac{S}{H^2}  \leq 3 \hat{P}.
 \end{align}
 
 It turns out that due to the presence of the magnetic field the estimate of the components of $\Sigma_{ij}$ will not be the same for all indices and it will be useful to consider the sum and difference of $\Sigma_{22}$ and $\Sigma_{33}$ which is also common in Bianchi cosmologies. Defining 
  \begin{align*}
 \Sigma_+ = \frac12 \left (\Sigma_{22}+\Sigma_{33} \right),\\
 \Sigma_- = \frac12 \left (\Sigma_{22}-\Sigma_{33} \right),
 \end{align*}
we obtain from \eqref{22}--\eqref{33}
  \begin{align}
 &\label{sigmaplus} \Sigma'_{+}= (q-2) \Sigma_{+}-\Sigma_{12}^2-\Sigma_{13}^2+\frac13 B^2 + \frac12 (\Pi_{22}+\Pi_{33} ),\\
&\label{sigmaminus}\Sigma'_{-}= (q-2) \Sigma_{-}-\Sigma_{12}^2+\Sigma_{13}^2+  \frac12 (\Pi_{22}-\Pi_{33} ) .
  \end{align}
  We will now proceed to prove the following result concerning massive particles.

\subsection{Main results}\label{main}

\begin{lem}\label{lem1}
Consider massive solutions to the Einstein-Vlasov system with a pure magnetic field and with Bianchi I symmetry. Let  $\Sigma_{12} (\tau_0)$, $\Sigma_{13} (\tau_0)$, $\Sigma_{23} (\tau_0)$, $\Sigma_{+} (\tau_0)$, $\Sigma_{-} (\tau_0)$, $H(\tau_0)$, $B(\tau_0)$ and $f(\tau_0)$ be  initial data of the Einstein-Vlasov system satisfying the constraints such that $f(\tau_0)$ has compact support in momentum space and $ \vert B(\tau_0) \vert $, $\vert \Sigma_{12} (\tau_0)\vert $, $\vert \Sigma_{13}  (\tau_0) \vert$, $\vert \Sigma_{23} (\tau_0)\vert $, $\vert \Sigma_{+} (\tau_0)\vert $, $\vert \Sigma_{-} (\tau_0)\vert $ and $\hat{P}(\tau_0)$ are non-zero and sufficiently small.
Then, there exists a small positive constant $\varepsilon$ such that the following estimates hold:
\begin{align}
\label{goalB} B &= O \left(\varepsilon e^{-\frac12\tau} \right),\\
\label{goalsigma} \Sigma_{ij} & = O \left(\varepsilon e^{-\frac32\tau} \right), \quad i\neq j\\
\label{goalsigmaminus} \Sigma_{-} & = O \left(\varepsilon e^{-\frac32\tau} \right), \\
\label{goalsigmaplus} \Sigma_{+} & = O \left(\varepsilon e^{-\tau} \right), \\
\label{goalP} \hat{P} &= O \left(\varepsilon e^{-2\tau} \right).
\end{align}
\end{lem}

\begin{proof}
The proof consists of two parts: in the first part we use a bootstrap argument to obtain certain decays, and then in the second part we improve the decay rates by removing the $\varepsilon$ from the decay rate exponents.

\begin{enumerate}
\item The first part will be a bootstrap argument.  Assume there exists an interval $[\tau_0, \tau)$ such that the following estimates hold:
\begin{align}
\label{bootB} \vert B \vert &\leq \varepsilon \exp \left[-\frac13 (\tau-\tau_0) \right],\\
\label{bootSigma2} \vert \Sigma_{ij}  \vert & \leq   \varepsilon \exp \left[-\frac78 (\tau-\tau_0) \right].
\end{align}

In what follows we will denote by $C$ some positive constant which might change from line to line.
Taking the supremum of \eqref{estimateP} and having in mind the definition \eqref{defPhat} we obtain with the bootstrap assumption \eqref{bootSigma2}
\begin{align}
\hat{P}(\tau) \leq \hat{P} (\tau_0)\exp \left(-2 + C\varepsilon \right)(\tau-\tau_0) , \label{estimateP2}
\end{align}
which implies decay for $\hat{P}(\tau)$. Since we assume that $\hat{P}(\tau_0)$ is small, $\hat{P}$ will remain small. Let us choose $\hat{P}(\tau_0)$ smaller than $\varepsilon$. From \eqref{H-1}, the bootstrap assumptions \eqref{bootB}--\eqref{bootSigma2}, the bound \eqref{boundSH2} and the definition \eqref{q} of $q$ we obtain
\begin{align}\label{estq}
\left\vert q - \frac12  \right\vert \leq  C \varepsilon.
\end{align}
As a result we obtain from \eqref{B} and the smallness assumptions
\begin{align*}
 -\frac12 - C \varepsilon \leq \frac{B'}{B} \leq -\frac12 + C \varepsilon,
\end{align*}
which integrated gives us
\begin{align}\label{estB}
\vert B(\tau_0) \vert e^{ \left( -\frac12 - C \varepsilon \right) (\tau-\tau_0)} \leq \vert B \vert \leq \vert B(\tau_0) \vert e^{ \left( -\frac12 + C \varepsilon \right) (\tau-\tau_0) } .
\end{align}
This is an improvement of the bootstrap assumption \eqref{bootB} choosing  $B(\tau_0)$ smaller than $\varepsilon$ and choosing  $\varepsilon$ to be small, and also an improvement of the first lower bound \eqref{blower}.

Now we want to improve \eqref{bootSigma2}.    From \eqref{Pi} we know that we can bound the terms $\Pi_{ij}$ by a constant times $\hat{P}$, which we already have estimated in \eqref{estimateP2}. We can estimate any quadratic term of  $\Sigma_{ij}$ via the bootstrap assumption \eqref{bootSigma2}. We also have an improved estimate for $B$ via \eqref{estB}.

For the rest of the variables we will use a contradiction argument. We will assume the opposite of the estimate we want to obtain in some interval obtaining a contradiction. Let us consider $\Sigma_{23}$. We want to show that \eqref{bootSigma2} can be improved for $\Sigma_{23}$ in that interval $[\tau_0, \tau)$ to
\begin{align}\label{improved}
 \vert \Sigma_{23}  \vert & \leq   \varepsilon \exp \left[ \left(-\frac32+\delta \right) (\tau-\tau_0) \right],
\end{align}
where $\delta$ is a small quantity such that $0<\delta<\frac58$ and does not depend on $\tau$. If $\Sigma_{23}$ satisfies \eqref{improved} we have improved the bootstrap assumption for $\Sigma_{23}$ in that interval $[\tau_0, \tau)$. Let us suppose the opposite namely that for any $0<\delta<\frac58$, there exist $\tau_1$ and $\tau_2$ such that $\tau_0 \leq \tau_1 < \tau_2 \leq \tau $ and
\begin{align}
\label{number} \vert \Sigma_{23} (\tau_*)  \vert & \leq   \varepsilon \exp \left[ \left(-\frac32+\delta\right) (\tau_*-\tau_0) \right], \quad t_* \in [\tau_0,\tau_1],\\
\label{contra}  \vert \Sigma_{23} (\tau_*)   \vert & \geq   \varepsilon \exp \left[ \left(-\frac32+\delta \right) (\tau_*-\tau_0) \right], \quad t_* \in [\tau_1,\tau_2).
\end{align}
Now let us consider the second interval. For that interval using \eqref{23} we have
\begin{align}\label{start}
\frac{\Sigma'_{23}}{\Sigma_{23}}= q-2- \frac{2\Sigma_{12}\Sigma_{13}}{\Sigma_{23}}+\frac{\Pi_{23}}{\Sigma_{23}}.
\end{align}
Using the bootstrap assumptions for $\Sigma_{12}$, $\Sigma_{13}$ which are \eqref{bootSigma2}, the estimate for $q$ \eqref{estq}, the estimate \eqref{estimateP2} together with \eqref{Pi} to estimate $\Pi_{23}$ and \eqref{contra} we obtain
\begin{align*}
\frac{\Sigma'_{23}}{\Sigma_{23}}  \leq -\frac32 + C  \varepsilon +  \frac{ C \varepsilon^2 e^{-\frac74 \tau_*} }{ \varepsilon  e^{\left(-\frac32 + \delta \right) \tau_*} }+ \frac{ C \hat{P}(\tau_0) e^{ \left(-2+C\varepsilon \right) \tau_*} }{ \varepsilon  e^{\left(-\frac32 + \delta \right) \tau_*} }.
\end{align*}
We can choose $\hat{P}(\tau_0)$ independently of $\varepsilon$ as small as we want. Choosing it smaller than $\varepsilon$ implies that for any small $\delta$ and small  $\varepsilon$ we can find a small $\delta_1$ which does not depend on $\delta$ such that
\begin{align*}
 \frac{\Sigma'_{23}}{\Sigma_{23}}  \leq -\frac32 + \delta_1.
\end{align*}
Integrating this inequality on $(\tau_1,\tau_*)$ with $t_* \in [\tau_1,\tau_2)$ we obtain
\begin{align*}
\Sigma_{23}(\tau_*) \leq \Sigma_{23}(\tau_1) \exp \left[ \left( -\frac32 + \delta_1 \right)(\tau_*-\tau_1) \right].
\end{align*}
Now we use \eqref{number} for $\tau_* = \tau_1$ to obtain
\begin{align*}
\Sigma_{23}(\tau_*) \leq \varepsilon  \exp \left[ \left(-\frac32+\delta\right) (\tau_1-\tau_0) \right] \exp \left[ \left( -\frac32 + \delta_1 \right)(\tau_*-\tau_1) \right].
\end{align*}
Since $\delta$ has been chosen arbitrarily and $\delta_1$ does not depend on $\delta$ we may assume that $\delta= \delta_1$ which implies
\begin{align*}
\Sigma_{23}(\tau_*) \leq \varepsilon  \exp \left[ \left(-\frac32+\delta\right) (\tau_*-\tau_0) \right],
\end{align*}
in the interval $\tau_* \in [\tau_1,\tau_2)$. Similarly starting from \eqref{start} one can obtain the corresponding estimate in the other direction, so that one can conclude
\begin{align*}
\vert \Sigma_{23}(\tau_*) \vert \leq \varepsilon  \exp \left[ \left(-\frac32+\delta\right) (\tau_*-\tau_0) \right],
\end{align*}
which contradicts \eqref{contra}.

We can use the same procedure integrating the evolution equations \eqref{22}--\eqref{13} for $\Sigma_{ij}$ and we obtain the following estimates for $\Sigma_{ij}$:
\begin{align}
&\label{est1}\vert \Sigma_{ij} (\tau) \vert \leq \vert \Sigma_{ij} (\tau_0)  \vert \exp \left[ \left(-\frac32 +\delta_1 \right)(\tau-\tau_0) \right], \quad i\neq j,\\
&\label{est2} \vert \Sigma_- (\tau) \vert \leq \vert \Sigma_- (\tau_0)  \vert \exp \left[ \left(-\frac32 +\delta_2 \right)(\tau-\tau_0)\right], \\
&\label{est3} \vert \Sigma_+ (\tau) \vert \leq \vert \Sigma_+ (\tau_0)  \vert \exp \left[ \left(-1 +\delta_3 \right)(\tau-\tau_0)\right], 
\end{align}
where $\delta_1$, $\delta_2$, $\delta_3$ are small quantities. We have improved the bootstrap assumption \eqref{bootSigma2} and have thus closed the bootstrap argument. Since the arguments used do not depend on $\tau$ we can conclude from the bootstrap argument that the estimates \eqref{estimateP2}, \eqref{estB} and \eqref{est1}--\eqref{est3} are in fact valid in the whole interval $[\tau_0, \infty)$.

\item In this part we want to get rid of the epsilon in the estimates obtained following \cite{LN2}. Using \eqref{ineqP} and the estimates \eqref{est1}--\eqref{est3} on $\Sigma_{ij}$ and \eqref{estimateP2} on $\hat{P}$ we obtain assuming that $\Sigma_{ij} (\tau_0)$ and $\hat{P}(\tau_0)$ are small the following:
\begin{align*}
(e^{2\tau} P)' = e^{2\tau}P' + 2 e^{2\tau} P \leq e^{2\tau} 2 \sqrt{\Sigma_{ij} \Sigma^{ij} }P \leq C \varepsilon e^{(-1 + C\varepsilon)\tau},
\end{align*}
which integrated gives us
\begin{align*}
P \leq  (C P(\tau_0) + \varepsilon) e^{-2\tau}.
\end{align*}
Taking the supremum, since $ \hat{P}(\tau_0)$ is  assumed to be small, gives the estimate for $\hat{P}$ which is \eqref{goalP}. Similarly using in addition of the assumptions on $\Sigma_{ij}$ and $\hat{P}$ the estimate of  \eqref{estB} we obtain with \eqref{B} that 
\begin{align*}
(B e^{\frac12 \tau} )' \leq C \varepsilon e^{(-1 + C\varepsilon)\tau},
\end{align*}
which integrated gives us our desired estimate for $B$. The same holds for $\Sigma_{ij}$ with $i\neq j$ where one obtains the inequality
\begin{align*}
(\Sigma_{ij} e^{\frac32 \tau} )' \leq C \varepsilon e^{(-1 + C\varepsilon)\tau}, \quad i \neq j,
\end{align*}
and integrating the desired result is $\Sigma_{ij} = O \left(\varepsilon e^{-\frac32\tau} \right)$ for $i \neq j$. Similarly for $\Sigma_-$. For $\Sigma_+$ however we only obtain
\begin{align*}
(\Sigma_{ij} e^{\frac32 \tau} )' \leq C \varepsilon e^{\frac12 \tau}, 
\end{align*}
due to the presence of the magnetic field. Integrating yields the desired result \eqref{goalsigmaplus} for $\Sigma_+$.

\end{enumerate}
\end{proof}

\noindent {\bf Remark.} Due to the lower bound \eqref{estB} for $B$, the estimate for $B$ obtained in Lemma 1 is in fact optimal. \medskip

With the estimates obtained we will obtain an estimate for the Hubble variable $H$ so that we then are able to express our estimates in terms of the time variable $t$. We will consider the case that $H(t_0)>0$.

\begin{thm}
Consider massive solutions to the Einstein-Vlasov system with a pure magnetic field and with Bianchi I symmetry. Let $\Sigma_{12} (\tau_0)$, $\Sigma_{13} (\tau_0)$, $\Sigma_{23} (\tau_0)$, $\Sigma_{+} (\tau_0)$, $\Sigma_{-} (\tau_0)$, $H(\tau_0)>0$, $B(\tau_0)$ and $f(\tau_0)$ be  initial data of the Einstein-Vlasov system satisfying the constraints such that $f(\tau_0)$ has compact support in momentum space and $ \vert B(\tau_0) \vert $, $\vert \Sigma_{12} (\tau_0)\vert $, $\vert \Sigma_{13}  (\tau_0) \vert$, $\vert \Sigma_{23} (\tau_0)\vert $, $\vert \Sigma_{+} (\tau_0)\vert $, $\vert \Sigma_{-} (\tau_0)\vert $ and $\hat{P}(\tau_0)$ are non-zero and sufficiently small. Then, there exists a small positive constant $\varepsilon$ such that the following estimates hold:
\begin{align*}
H &=  \frac23 t^{-1} (1+ O(\varepsilon t^{-\frac23} )),\\
 q & = \frac12 + O(\varepsilon t^{-\frac23}),\\
 B &= O \left(\varepsilon t^{-\frac13} \right),\\
 \Sigma_{ij} & = O \left(\varepsilon t^{-1} \right), \quad i \neq j\\
  \Sigma_{-} & = O \left(\varepsilon t^{-1} \right),\\
    \Sigma_{+} & = O \left(\varepsilon t^{-\frac23} \right),\\
 \hat{P} &= O \left(\varepsilon t^{-\frac43} \right).
\end{align*}
\end{thm}
\begin{proof}
Given initial data one can always add an arbitrary constant to the time origin and we choose our time coordinate and time origin such that 
\begin{align}\label{timeoriginchoice}
t_0= \frac{2}{3H(t_0)},
\end{align}
which is a positive number since $H(t_0)$ is positive.
With the estimates obtained in Lemma \ref{lem1} we obtain from \eqref{H-1} the following bound
\begin{align*}
\frac32 \leq - \frac{\dot{H}}{H^2} \leq \frac32 + C\varepsilon,
\end{align*}
which integrated and using \eqref{timeoriginchoice}
\begin{align*}
\frac32 t \leq \frac{1}{H} \leq \left( \frac32 + C \varepsilon \right) t - \varepsilon t_0  \leq \left( \frac32 + C \varepsilon \right) t.
\end{align*}
Using this bound in \eqref{dtdtau} we have
\begin{align*}
\frac32 t \leq \frac{dt}{d\tau} \leq \left( \frac32 + C \varepsilon \right) t , 
\end{align*}
which integrated and making some computations gives
\begin{align}\label{ttau}
 C_1 t^{-\frac23} \leq e^{-\tau} \leq C_2 t^{-\frac23 + C \varepsilon},
\end{align}
where $C_1$ and $C_2$ are some positive constants.
Again from \eqref{H-1} we obtain
\begin{align} \label{Hintermediate}
H = \frac{1}{\frac32 t + I } =\frac23 t^{-1} \frac{1}{1+\frac23 t^{-1} I}, 
\end{align}
with 
\begin{align}\label{I}
I = \int_{t_0}^t  \left[ \frac{1}{4} \Sigma_{ab}\Sigma^{ab} + \frac16 \left (\frac{S}{H^2}+ \frac12 B^2 \right) \right] ds.
\end{align}
Using the estimates of Lemma \ref{lem1} with \eqref{ttau} we obtain from \eqref{I}
\begin{align*}
0 \leq I \leq C \varepsilon t^{\frac13+C\varepsilon}.
\end{align*}
Using this in \eqref{Hintermediate} we obtain the estimate
\begin{align}
H =  \frac23 t^{-1} (1+ O(\varepsilon  t^{-\frac23+C\varepsilon} )).
\end{align}
Using this estimate again in \eqref{dtdtau} we can improve \eqref{ttau} getting rid of the $\varepsilon$, i.e.~
\begin{align}\label{ttauagain}
 e^{-\tau} = O( t^{-\frac23}).
\end{align}
Making another loop we obtain our desired estimate for $H$ and we can express the estimates of Lemma \ref{lem1} in terms of $t$.
\end{proof}

\section{The positive cosmological constant case} \label{lambda}

In presence of a cosmological constant $\Lambda$, which we will assume to be positive, the Einstein equations are in a general frame
\begin{align*}
G_{\alpha\beta}=T_{\alpha\beta} - \Lambda g_{\alpha \beta},
\end{align*}
where we have put the cosmological constant on the right hand side, so that we can consider the term with a cosmological constant as an extra term in the energy-momentum tensor, like $T^{\Lambda}_{\alpha \beta}=  - \Lambda g_{\alpha}$. Considering again an orthonormal frame we have that the contributions to the energy-momentum tensors in the Einstein equations by this term will be
\begin{align*}
\rho^{\Lambda} = \Lambda, \quad \mathcal{P}^{\Lambda} = - \Lambda, \quad q^{\Lambda}_a =0 , \quad \pi_{ab}^{\Lambda}= 0.
\end{align*}
This means that the only Einstein equations in \eqref{evH}--\eqref{momentumconstraint} which will be modified are \eqref{evH} and \eqref{constraintrho} which are now as follows:
\begin{align}
\label{1Lambda} &\dot{H}= -H^2 - \frac13  \sigma_{ab}\sigma^{ab}- \frac16 (\rho^{\mathrm{Vl}}+S+h^2  -2 \Lambda),\\
\label{2Lambda} &  \rho^{\mathrm{Vl}}+\frac12 h^2+\Lambda = 3H^2-\frac12 \sigma_{ab}\sigma^{ab}.
\end{align}

We obtain the following result
\begin{thm}
Consider solutions to the Einstein-Vlasov system with a pure magnetic field, a positive cosmological constant $\Lambda$ and with Bianchi I symmetry. Let  $\Sigma_{ij} (t_0)$, $H(t_0)>0$, $B(t_0)$ and $f(t_0)$ be initial data of the Einstein-Vlasov system satisfying the constraints such that $f(t_0)$ has compact support in momentum space (away from the origin in the massless case). Then, the following estimates hold:
\begin{align}
H &=  \sqrt{\frac{\Lambda}{3}} + O \left( e^{-2\sqrt{\frac{\Lambda}{3}}t}\right),\\
\Sigma_{ab} \Sigma^{ab} & = O \left(e^{-2\sqrt{\frac{\Lambda}{3}}t}\right), \\
B & = O \left(e^{-\sqrt{\frac{\Lambda}{3}}t}\right),\\
q & = -1 +  O \left(e^{-2\sqrt{\frac{\Lambda}{3}}t}\right).
\end{align}
\end{thm}

\begin{proof}
From the equations \eqref{1Lambda}--\eqref{2Lambda}, since $ \sigma_{ab}\sigma^{ab}$, $S$, $\rho^{\mathrm{Vl}}$ and $h^2$ are positive, we obtain the inequalities
\begin{align}
 &-\dot{H}=  H^2 +\frac13  \sigma_{ab}\sigma^{ab} + \frac16 (\rho^{\mathrm{Vl}}+S+h^2)  -\frac13 \Lambda \geq H^2  -\frac13 \Lambda  ,\\
\label{constraintlambda}&  \rho^{\mathrm{Vl}}+\frac12 h^2+\frac12 \sigma_{ab}\sigma^{ab} = 3H^2- \Lambda  \geq 0.
\end{align}
Using the second inequality in the first we obtain
\begin{align}\label{ineqH}
\dot{H} \leq - H^2 + \frac13 \Lambda \leq 0,
\end{align}
which means that $H$ is decreasing and also implies 
\begin{align}\label{boundH2}
H^2 \geq \frac13 \Lambda.
\end{align}
Thus if $H(t_0)>0$, this will always be the case. In order to obtain an estimate of $H$ we integrate \eqref{ineqH}. If  $- H^2 + \frac13 \Lambda =0$ we have 
\begin{align}\label{Hequality}
H = \sqrt{\frac{\Lambda}{3}}.
\end{align}
If  $- H^2 + \frac13 \Lambda \neq 0$ we obtain from \eqref{ineqH}
\begin{align*}
\frac{\dot{H}}{H^2 - \frac13 \Lambda } \leq -1,
\end{align*}
which can be expressed as
\begin{align*}
\frac{d}{dt} \left[  \ln  \frac{H - \sqrt{\frac{\Lambda}{3}}}{H + \sqrt{\frac{\Lambda}{3}}} \right ] \leq -2\sqrt{\frac{\Lambda}{3}}.
\end{align*}
Integrating the above we have that
\begin{align}\label{Hresult}
H = \sqrt{\frac{\Lambda}{3}} + C e^{-2\sqrt{\frac{\Lambda}{3}}t} ,
\end{align}
since $H$ is decreasing. Combining \eqref{Hequality} and \eqref{Hresult} we have the desired result. From \eqref{constraintlambda} we obtain that $h^2$ and $\sigma_{ab}\sigma^{ab}$ are bounded by $O(e^{-2\sqrt{\frac{\Lambda}{3}}t})$. Since $H^2 $ is bounded from below by a constant due to \eqref{boundH2} we obtain the desired estimates for $\Sigma_{ab}\Sigma^{ab}$ and $B$. Using the definition \eqref{q} of $q$ with \eqref{1Lambda} and having in mind \eqref{energyconditions} we obtain the estimate for $q$.
\end{proof}

\section{Conclusions and Outlook}

 We have considered small solutions to the Einstein-Vlasov system with Bianchi I symmetry and a pure magnetic field because Bianchi I magnetic space-times are of physical interest and might help to solve the Hubble tension \cite{Akarsu,CKPM,DV, Jedamzik}. We have shown that these solutions have a dust-like behaviour at late times and have obtained the decay rates of the main variables involved generalising \cite{EN}. 
 
 The dust case with a magnetic field had already been studied in \cite{LB} for Bianchi I space-times. It turns out that for a perfect fluid with a magnetic field the Bianchi I case \cite{LB} is somehow more complicated than the Bianchi VI$_0$ case which had been studied earlier \cite{LKW}, because diagonalising the shear tensor leads to complications in the former case (cf.\ \cite{LB} for details) and it was thus preferable to deal with a non-vanishing local angular velocity $\Omega^a$ of the spatial frame. 
 
 The framework developed in \cite{LB} works nicely in the Vlasov case because what really matters is the dispersion of the momenta which still decreases if the universe is expanding and will not depend on the local angular velocity as we have demonstrated. This shows that the techniques developed in \cite{EN} which also have been applied to higher Bianchi A types \cite{LN2, E4,E3} and a Bianchi B type \cite{E6}, massless particles \cite{BFH, B, LNS, LNT2} and the Boltzmann equation \cite{LN} are robust. Note that in contrast to the previous work we have worked entirely using an orthonormal frame. This has made not only the generalisation of \cite{LB} easier, but will also be useful for a generalisation of the Boltzmann case \cite{Ayissi,LN,LN3} because in the case of an orthonormal frame the collision operator has the same form as for the Minkowski case. The restrictions on the collision kernel will depend on the assumption or not of a cosmological constant. 
 
 It would also be of interest to generalise the present work to higher Bianchi types with a magnetic field \cite{C, HW, LB, LB2, LKW,LRT, Weaver}  both towards the future or the direction of the singularity and to show non-linear stability of not necessarily symmetric solutions of the Einstein-Vlasov system with a magnetic field \cite{AR, E5, Ring, Svedberg}.

\section*{Acknowledgements}
The authors thank John Stalker from Trinity College Dublin and Paul Tod from the University of Oxford for helpful discussions.

\end{document}